\documentclass[12pt,a4paper] {article}

\usepackage{latexsym,amsmath,enumerate,amssymb,amsbsy,amsthm}
 
\usepackage{enumerate,amssymb,amsmath,amsthm,mathrsfs,latexsym,amsfonts,txfonts,graphicx,multicol,color, colortbl,mathrsfs, multirow, verbatim,wasysym}

\definecolor{LightCyan}{rgb}{.9,.9,.9}

\theoremstyle{plain}
\newtheorem{thm}{Theorem}[section]
\newtheorem{lem}[thm]{Lemma}

\newtheorem{cor}[thm]{Corollary}
\theoremstyle{definition}

\newtheorem{prop}[thm]{Proposition}
\newtheorem{ex}[thm]{Example}

\usepackage[table]{xcolor}

\newcommand{\Z}{\ensuremath{\mathbb{Z}_4}}
\newcommand{\F}{\ensuremath{\mathbb{F}_{2}}}
\newcommand{\al}{\ensuremath{\alpha}}
\newcommand{\be}{\ensuremath{\beta}}
\newcommand{\gam}{\ensuremath{\gamma}}

\DeclareMathOperator{\ord}{ord}
\DeclareMathOperator{\Hull}{Hull}
\DeclareMathOperator{\lcm}{lcm}

\renewcommand\footnotemark{}
\begin{document}
 
\title{Hulls of Cyclic Codes over $\mathbb{Z}_4$}
 \author{Somphong Jitman, Ekkasit Sangwisut, and  Patanee Udomkavanich}

 \thanks{Somphong Jitman is with the Department of Mathematics, Faculty of Science, Silpakorn University, Nakhon Pathom 73000, Thailand}
 
  \thanks{Ekkasit Sangwisut (Corresponding author: ekkasit.sangwisut@gmail.com) is with the 
 Department of Mathematics and Statistics, Faculty of Science, Thaksin University, Phattalung 93110, Thailand} 

\thanks{Patanee Udomkavanich is with the Department of Mathematics and Computer Science, Faculty of Science,  Chulalongkorn University,   Bangkok 10330,  Thailand} 
%
%
%
%
%
%
\maketitle
\begin{abstract} 
 The hulls  of  linear  and cyclic codes over finite fields  have been of interest  and  extensively studied due to their  wide applications.  In this paper,   the hulls of cyclic codes of length $n$ over the   ring $\mathbb{Z}_4$ have been focused on. Their  characterization has  been  established  in terms of the generators viewed as ideals in the quotient ring $\mathbb{Z}_4[x]/\langle x^n-1\rangle$. An algorithm for computing the types  of the hulls of cyclic codes of   arbitrary  odd length over  $\Z$ has been given. 
 The average $2$-dimension $E(n)$  of the hulls of cyclic codes of  odd length $n$ over $\Z$ has been established.  A general  formula for $E(n)$ has been provided  together with its upper and lower bounds. It turns out that $E(n)$ grows the same rate as $n$. \bigskip
\\{\it Keywords}:  hulls, cyclic codes, reciprocal polynomials, 	average $2$-dimension \\ {\it MSC2010}:  11T71, 11T60, 94B05
\end{abstract}

\section{Introduction}
The hull of  a linear code, the intersection of the code and its dual, has been first  introduced in \cite{Assmus90}  to classify finite projective planes.     Properties of hulls have been extensively studied since  the hull dimension is key to determine the complexity of   algorithms for investigating permutation equivalence of two linear codes and calculating the automorphism of a fixed linear code  given in \cite{Leon82, Leon97, Leon91,  Sendrier997, Sendrier00, Sendrier01}. Precisely,    most of the algorithms do  not work if the  hull  dimension is large. 

Recently, the hulls of linear codes  have been applied  in the construction of good entanglement-assisted quantum error correcting codes in \cite{GJG2016}. 
Therefore, the study of the hulls and the hull dimensions of linear codes over finite fields  has become  of interest.  The number of distinct linear codes of length $n$ over  a finite field whose hulls share  a given dimension has been established in \cite{Sendrier97} together with the average hull dimension  of linear codes of length $n$ over a finite field. In \cite{Skersys03}, this study has been extended to the class of cyclic codes over finite fields and the average hull dimension of   cyclic codes  has been determined.    Later, the hull dimensions of cyclic and negacyclic codes and  the number of cyclic  codes whose  hulls share a given dimension have been established in \cite{Sangwisut}. The average hull  dimension   of constacyclic codes over finite fields  have been given in \cite{ Sangwisut16, Sangwisut17, Sangwisut18}.

In the early history of  coding theory, codes were usually taken over finite fields. In the
last three decades,  interest has been shown in linear codes over rings.   In an important work \cite{Calderbank93, Hammons94},
it has been  shown that the Kerdock codes,   Preparata codes and Delsarte-Goethals codes can be
obtained through the Gray images of linear codes over $\mathbb{Z}_4$.  {Some properties and applications  of the hulls of linear codes over finite rings have been introduced and studied in \cite[Chapter 5]{D2017}  and \cite{HL2000}.  Most of the  study of  the hulls of codes over rings have been done in the two special cases where the hull is trivial (complementary dual code)  in   \cite{LL2015}  and the hull equals the code itself (self-orthogonal code) in \cite{QZK2015} and \cite{SKS2018}.  
    It is therefore of interest to  investigate   properties of the hulls of linear and cyclic codes  over rings for arbitrary cases.}  In this paper, we focus on  the hulls of cyclic codes of odd length $n$  over $\Z$. 
The characterization of the hulls of cyclic codes over $\Z$ is given in terms of their generators viewed as  ideals in $\mathbb{Z}_4[x]/\langle x^n-1\rangle$. Based on this  characterization, the types of the hulls of cyclic codes of odd length $n$ over $\Z$ are  determined. Furthermore, the average $2$-dimension and its upper and lower bounds are  derived.

The paper is organized as follows. In Section 2,     basic concepts and key results on cyclic codes over $\Z$  are recalled.  In Section 3, the characterization of    the hull of cyclic codes of odd length $n$ over $\Z$  is given in terms of their generators. Subsequently, the  types of the hulls of  such codes  are determined. 
The formula for the average $2$-dimension of the hull of cyclic codes of  odd length $n$ over $\Z$ is derived   in Section 4.
In Section 5, upper  and lower bounds on $E(n)$ are given  together with asymptotic behaviors of $E(n)$.

\section{Preliminaries}	

In this section, definitions and   preliminary results required in the study of the hulls of cyclic codes  over $\Z$ are recalled.  Precisely, properties of linear codes,  hulls of codes,   cyclic codes and polynomials  over $\Z$ are  discussed. 

\subsection{Linear Codes and Hulls over $\Z$}
A \textit{linear code} $C$ of length $n$ over $\Z$ is defined to be  a submodule of the $\Z$-module $\Z^n$.  As a linear code  $C$  of length $n$ over  $\mathbb{Z}_4$ can be viewed as a vector space over $\mathbb{F}_2$, the concept of   \textit{$2$-dimension} of $C$  was introduced  in  \cite{VHR1996} to be  $\dim_2( C)=\log_2(|C|)$.   Elements $\boldsymbol{u}=(u_0,u_1,\ldots, u_{n-1})$ and $\boldsymbol{v}=(v_0,v_1,\ldots ,v_{n-1})$  in $\Z^n$ are said to be \textit{orthogonal}   if and only if $\sum_{i=0}^{n-1}u_iv_i=0.$  Subsets $U$ and $V$ of $\Z^n$ are said to be {\em orthogonal} if $\boldsymbol{u}$ is orthogonal to $\boldsymbol{v}$ for all $\boldsymbol{u}\in U$ and $\boldsymbol{v}\in V$.
The dual of a linear code $C$ of length $n$ over $\Z$ is defined to be the linear code
\[C^\perp=\left\{\boldsymbol{u}\in\Z^n\mid\boldsymbol{u} \text{~is orthogonal to~} \boldsymbol{c}  \text{~for all~} \boldsymbol{c}\in C \right\}.\]
The \textit{hull} of a linear code $C$ is defined to be $$\Hull(C)=C\cap C^\perp.$$

\subsection{Cyclic Codes over $\Z$}
A linear code of length $n$ over $\Z$ is said to be \textit{cyclic} if $(c_{n-1}, c_0,\ldots, c_{n-2})\in C$ for all $(c_0, c_1,\ldots, c_{n-1}) \in C$. A vector   $\boldsymbol{u}=(u_0,u_1,\ldots, u_{n-1})$  in $\mathbb{Z}_4^n$ can be represented as its {\em corresponding polynomial} $u(x)=u_0+u_1x+\dots+u_{n-1}x^{n-1}$ in  $\Z[x]$.  It is well known that each cyclic code $C$ of length $n$ over $\Z$ can be viewed as an ideal of the quotient ring $R_n=\Z[x]/\langle x^n-1\rangle$ (see \cite[Chapter 12]{Huffman}).  Moreover, if $n$ is odd, 
the corresponding ideal of a cyclic code $C$ has generators of the form
$$\left\langle f(x)g(x), 2f(x)h(x)\right\rangle=\left\langle f(x)g(x), 2f(x)\right\rangle,$$
where  $f(x), g(x), h(x)$ are unique monic pairwise coprime polynomials such that  $x^n-1=f(x)g(x)h(x)$ (see \cite[Theorem 12.3.13]{Huffman}).  
Furthermore, $|C|=4^{\deg h(x)}2^{\deg g(x)}$  and $C$  is said to be of \textit{type} $4^{\deg h(x)}2^{\deg g(x)}$. In this case, the \textit{$2$-dimension} of $C$  is   $\dim_2( C)=\log_2(|C|)=2\deg h(x)+\deg g(x)$.

Let $f(x)=a_0+a_1x+\cdots+a_{k-1}x^{k-1}+ x^k\in \Z[x]$ (resp.,~$\F[x]$)  be a  monic polynomial  such that $a_0$ is a unit in $\Z$ (resp.,~$\F$). The \textit{reciprocal polynomial} of $f(x)$ is defined to be 
$$f^*(x)=a_0^{-1}x^{\deg f(x)}f\left(\frac{1}{x}\right).$$
Clearly,  $(f^*)^*(x)=f(x)$. Therefore, there are two types of monic  polynomials in $\mathbb{Z}_4[x]$ (resp., $\mathbb{F}_2[x]$) whose constant terms are units. A polynomial $f(x)$ is called \textit{self-reciprocal} if $f(x)=f^*(x)$. Otherwise, $f(x)$ and $f^*(x)$ are called a \textit{reciprocal polynomial pair}. Note that $f(x)g(x)h(x)=x^n-1=\left(x^n-1\right)^*=f^*(x)g^*(x)h^*(x)$.

For a cyclic code $C$  of length $n$ over $\Z$   generated by $\left\langle f(x)g(x), 2f(x)\right\rangle$, the dual $C^\perp$ is generated by \begin{align} \label{cperp} \left\langle h^*(x)g^*(x), 2h^*(x)f^*(x)\right\rangle=\left\langle h^*(x)g^*(x), 2h^*(x)\right\rangle\end{align} (see \cite[Theorem 12.3.20]{Huffman}).

For  a positive integer $n$,  let $\mathcal{C}(n, 4)$ denote the set of all cyclic codes of length $n$ over $\Z$. The {\em average $2$-dimension} of the hull of cyclic codes of length $n$ over $\Z$ is defined to be 
\begin{align*}  E(n)=\sum_{C\in\mathcal{C}(n, 4)}\frac{\dim_2(\Hull (C))}{|\mathcal{C}(n, 4)|}.
\end{align*}    
Properties of the  average $2$-dimension  $E(n)$ of the hull of cyclic codes of length $n$ over $\Z$   are studied in Sections 4--5. 
\subsection{Factorization of $x^n-1$ over $\Z$}

In this subsection,  the factorization  of $x^n-1$ over $\Z$ for odd positive integers $n$ is recalled. 
Let  $\mu:\Z[x]\rightarrow \F[x]$  be a map defined by $\mu(0)=0=\mu(2)$, $\mu(1)=1=\mu(3)$ and $\mu(x)=x$.  It follows that  $\mu(x^n-1)=x^n-1$. 

For coprime positive integers $i$ and $j$, let  $\ord_j(i)$ denote the multiplicative order of   $i$  modulo $j$. 
Let $N_2:=\left\{\ell\geq 1 : \ell \text{~~divides~~} 2^i+1 \text{ for some positive integer }i \right\}$. From \cite{Sangwisut}, the factorization of $x^n-1$ in $\F[x]$ is of the form 
\begin{align*}
x^n-1=&\prod_{j| n, j\in N_2}\left(\prod_{i=1}^{\gam(j)}h_{ij}(x)\right) \prod_{j| n, j\not\in N_2}\left(\prod_{i=1}^{\be(j)}k_{ij}(x)k^*_{ij}(x)\right),
\end{align*} 
where $$\gam(j)=\frac{\phi(j)}{\ord_j(2)},~~~\be(j)=\frac{\phi(j)}{2\ord_j(2)},$$
$k_{ij}(x)$ and $k_{ij}^*(x)$ form a monic irreducible  reciprocal polynomial pair of degree $\ord_j(2)$ and $h_{ij}(x)$ is a monic irreducible self-reciprocal polynomial of degree $\ord_j(2)$. By  Hensel's lift (see
\cite[Theorem 12.3.7]{Huffman}),  the factorization of $x^n-1$ in $\Z[x]$ is
\begin{align}\label{xn-1}
x^n-1=&\prod_{j| n, j\in N_2}\left(\prod_{i=1}^{\gam(j)}g_{ij}(x)\right) \prod_{j| n, j\not\in N_2}\left(\prod_{i=1}^{\be(j)}f_{ij}(x)f^*_{ij}(x)\right)
\\
\label{st} =&\prod_{i=1}^{\mathtt{s}}g_i(x)\prod_{j=1}^{\mathtt{t}}f_j(x)f_j^*(x),
\end{align}  where $f_{ij}(x), f_{ij}^*(x)$  form a monic basic irreducible  reciprocal polynomial pair and $g_{ij}(x)$ is a monic basic irreducible  self-reciprocal polynomial, \begin{align}
\label{eqs}
\mathtt{s}=\sum_{j\mid n, j\in N_2}\frac{\phi(j)}{\ord_j(2)}\end{align} is the number of  monic basic irreducible  self-reciprocal polynomial in the factorization of $x^n-1$, and \begin{align} \label{eqt} \mathtt{t}=\sum_{j\mid n, j\not\in N_2}\frac{\phi(j)}{2\ord_j(2)}\end{align} is the number of a monic basic irreducible  reciprocal polynomial pair in the factorization of $x^n-1$.
Moreover, $\mu(g_{ij}(x))=h_{ij}(x), \mu(f_{ij}(x))=k_{ij}(x)$ and $\mu(f_{ij}^*(x))=k_{ij}^*(x)$.

Let $B_n=\deg\displaystyle\prod_{j| n, j\in N_2}\left(\prod_{i=1}^{\gam(j)}g_{ij}(x)\right)$.  Then
\begin{align}\label{Bn}
B_n=\deg\prod_{j| n, j\in N_2}\left(\prod_{i=1}^{\gam(j)}g_{ij}(x)\right)=\sum_{j| n, j\in N_2}\frac{\phi(j)}{\ord_j(2)}\cdot\ord_j(2)=\sum_{j| n, j\in N_2}\phi(j).
\end{align}
The number $B_n$ plays an important role in the study  of the average $2$-dimension of the hull of cyclic codes over $\Z$  in Sections 4-5.

\section{Hulls of Cyclic Codes       over $\mathbb{Z}_4$}
In this section,   properties of the hulls of cyclic codes  of arbitrary odd lengths over $\mathbb{Z}_4$ are focused on.

From now on,  assume that $n$ is an odd positive integer. The characterization of the hulls of cyclic codes  of     length  $n$ over $\mathbb{Z}_4$ is given  in terms of their generators in Subsection 3.1. Subsequently, the types  of the hulls of cyclic codes  of   length  $n$  over $\mathbb{Z}_4$ are determined Subsection 3.2.

\subsection{Characterization of the Hulls of Cyclic Codes    over $\mathbb{Z}_4$}      Here, we focus on  algebraic structures of the hulls of cyclic codes of odd length $n$  over $\mathbb{Z}_4$. The following lemma is useful in the study of  their generators. 

\begin{lem}[{{\cite[Theorem 12.3.18]{Huffman}}}] \label{lem11}
    Let $\boldsymbol{u}=(u_0,u_1,\ldots, u_{n-1})$ and $\boldsymbol{v}=(v_0,v_1,\ldots, v_{n-1})$ be vectors in $\Z^n$ with corresponding polynomial $u(x)$ and $v(x)$, respectively. Then $\boldsymbol{u}$  is orthogonal to $\boldsymbol{v}$ and all its shifts if and only if $u(x)v^*(x)=0$ in $\Z[x]/\langle x^n-1\rangle$.
\end{lem}
The  generators of the hull   of  a cyclic code is  determined  as follows. 
\begin{thm}\label{genhull}
    Let $C$ be a cyclic code of odd length $n$ over $\Z$ generated by $$\left\langle f(x)g(x), 2f(x)\right\rangle,$$   where $x^n-1=f(x)g(x)h(x) $ and $f(x)$, $g(x)$ and  $h(x) $ are pairwise coprime. Then $\Hull(C)$ is  generated by 
    \begin{align*} 
    \left\langle\lcm\left(f(x)g(x), h^*(x)g^*(x)\right), 2\lcm\left(f(x),  h^*(x)\right)\right\rangle.
    \end{align*}
    Furthermore, $\Hull(C)$ is of  type $4^{\deg H(x)}2^{\deg G(x)}$, where \[H(x)={\gcd(h(x), f^*(x))} \text{ and } G(x)={{\frac{x^n-1}{\gcd(h(x), f^*(x))\cdot\lcm(f(x), h^*(x))}}}.\]
\end{thm}
\begin{proof}{From Eq  \eqref{cperp}, note that 
        $C^\perp$  is generated by $$\left\langle h^*(x)g^*(x), 2h^*(x)\right\rangle.$$ 
        Let $C'$ be a cyclic code of length $n$ over $\Z$ whose generators  are of the form \[\left\langle F(x)G(x), 2F(x)\right\rangle,\] where 
        \[F(x)=\lcm(f(x), h^*(x)),\]
        \[G(x)=\frac{\lcm(f(x)g(x), h^*(x)g^*(x))}{\lcm(f(x), h^*(x))}=\frac{x^n-1}{\gcd(h(x), f^*(x))\cdot\lcm(f(x), h^*(x))}\]
        and 
        \[H(x)=\frac{x^n-1}{\lcm(f(x)g(x), h^*(x)g^*(x))}=\gcd(h(x), f^*(x)).\]
        It is not difficult  to see that  $x^n-1=F(x)G(x)H(x)$  and the polynomials  $F(x)$, $G(x)$ and $H(x)$ are pairwise coprime. 
        Since $\left\langle F(x)G(x), 2F(x)\right\rangle\subseteq \left\langle f(x)g(x), 2f(x))\right\rangle$ and $\left\langle F(x)G(x), 2F(x)\right\rangle\subseteq \left\langle h^*(x)g^*(x), 2h^*(x)\right\rangle$, we have $C'\subseteq \Hull(C)$.
    }
    
    Next, we show that $ \Hull(C)\subseteq C'$. Since $\Hull(C)$ is a cyclic code of length $n$ over $\Z$, assume that $\Hull(C)$ has generators of the form  $\left\langle A(x)B(x), 2A(x)\right\rangle$ where $x^n-1=A(x)B(x)C(x)$ and  the polynomials $A(x), B(x)$ and $ C(x)$ are pairwise co-prime. Since $\Hull(C)\subseteq C^\perp $ is orthogonal to $C$, by Lemma \ref{lem11},  we have 
    $$A(x)B(x)\cdot 2f^*(x)=0~~~ \text{and}~~~ 2A(x)\cdot f^*(x)g^*(x)=0$$
    which imply that  $h^*(x)g^*(x)|A(x)B(x)$ and $h^*(x)|A(x)$. 
    
    Similarly,    $\Hull(C)\subseteq C$ is orthogonal to $C^\perp$  which implies that 
    $$A(x)B(x)\cdot 2h(x)=0~~~\text{and}~~~2A(x)\cdot h(x)g(x)=0$$
    by Lemma \ref{lem11}. It follows that $f(x)g(x)|A(x)B(x)$ and $f(x)| A(x)$.

    Consequently, $\lcm(f(x)g(x), h^*(x)g^*(x))|A(x)B(x)$ \ and \ $\lcm(h^*(x), f(x))|A(x)$   which~imply that  $F(x)G(x)| A(x)B(x)$ and $F(x)| A(x)$. Hence,  $\Hull(C)\subseteq C'$. Therefore,  $\Hull(C)= C'$ as desired. 
\end{proof}

An illustrative example of   Theorem \ref{genhull} is given as follows. 
\begin{ex}
    In  $\Z[x]$, $x^7-1=(x-1)(x^3+2x^2+x-1)(x^3-x^2+2x-1)$ is the factorization of $x^7-1$ into a product of monic basic  irreducible polynomials.     
    Let $C$ be the  cyclic code of length $7$ over $\Z$  generated by 
    $$\langle f(x)g(x), 2f(x) \rangle=\langle(x^3+2x^2+x-1)(x^3-x^2+2x-1), 2(x^3+2x^2+x-1)\rangle$$
    where $f(x)=x^3+2x^2+x-1, g(x)=x^3-x^2+2x-1$ and $h(x)=x-1$. Moreover, $f^*(x)=g(x)$ and $h^*(x)=h(x)$. From Eq \eqref{cperp},  $C^\perp$ is of the form  $$\langle h^*(x)g^*(x), 2h^*(x)\rangle=\langle(x-1)(x^3+2x^2+x-1), 2(x-1) \rangle.$$ 
    By Theorem \ref{genhull}, $\Hull(C)$ is  of the form 
    $$ \left\langle\lcm\left(f(x)g(x), h^*(x)g^*(x)\right), 2\lcm\left(f(x),  h^*(x)\right)\right\rangle=\langle 2(x-1)(x^3+2x^2+x-1)\rangle.$$   
\end{ex}

\subsection{Characterization of Cyclic Codes of  the same Hull}
For a given cyclic code $D$ of odd length $n$ over $\Z$, the cyclic codes of odd length $n$ over $\Z$ whose hulls equal $D$ are determined in this subsection.
\begin{thm}
    Let $D$ be a cyclic code of odd length $n$ over $\Z$ generated by
    \begin{align}
    \notag\left\langle \prod_{j| n, j\in N_2}\prod_{i=1}^{\gam(j)}g_{ij}(x)^{A_{ij}}\prod_{j| n, j\not\in N_2}\prod_{i=1}^{\be(j)}f_{ij}(x)^{B_{ij}}f_{ij}^*(x)^{C_{ij}},\right.\hspace{5cm}\\
    \label{eq10}\left. 2\prod_{j| n, j\in N_2}\prod_{i=1}^{\gam(j)}g_{ij}(x)^{D_{ij}}\prod_{j| n, j\not\in N_2}\prod_{i=1}^{\be(j)}f_{ij}(x)^{E_{ij}}f_{ij}^*(x)^{F_{ij}} \right\rangle,
    \end{align}
    where $g_{ij}(x)$ and $f_{ij}(x)$ are given in Eq \eqref{xn-1} and $A_{ij}, B_{ij}, C_{ij}, D_{ij}, E_{ij}, F_{ij} \in \{0,1\}$. 
    Then the cyclic codes of length $n$ over $\Z$ whose hulls equal $D$ are generated by
    \begin{align*}
    \left\langle \prod_{j| n, j\in N_2}\prod_{i=1}^{\gam(j)}g_{ij}(x)^{u_{ij}+b_{ij}}\prod_{j| n, j\not\in N_2}\prod_{i=1}^{\be(j)}f_{ij}(x)^{v_{ij}+z_{ij}}f_{ij}^*(x)^{w_{ij}+d_{ij}},\right.\hspace{4cm}\\
    \left. 2\prod_{j| n, j\in N_2}\prod_{i=1}^{\gam(j)}g_{ij}(x)^{u_{ij}}\prod_{j| n, j\not\in N_2}\prod_{i=1}^{\be(j)}f_{ij}(x)^{v_{ij}}f_{ij}^*(x)^{w_{ij}} \right\rangle, 
    \end{align*}
    where $(u_{ij}, b_{ij})\in\begin{cases} 
    \left\{(0, 1)\right\}& \text{if~~  } D_{ij}=0,\\
    \left\{(0, 0), (1, 0)\right\}& \text{if~~  } D_{ij}=1,
    \end{cases}$\\
    and \\
    $(v_{ij}, z_{ij}, w_{ij}, d_{ij})\in\begin{cases}
    \left\{(0, 0, 0, 0), (1, 0, 1, 0)\right\} &\text{if~~ } (A_{ij}, B_{ij}, C_{ij}, E_{ij}, F_{ij})=(1, 1, 1, 1, 1),\\
    \left\{(0, 1, 1, 0), (0, 0, 0, 1)\right\} &\text{if~~ } (A_{ij}, B_{ij}, C_{ij}, E_{ij}, F_{ij})=(1, 1, 1, 0, 1),\\
    \left\{(0, 1, 0, 0), (1, 0, 0, 1)\right\} &\text{if~~ } (A_{ij}, B_{ij}, C_{ij}, E_{ij}, F_{ij})=(1, 1, 1, 1, 0),\\
    \left\{(1, 0, 0, 0)\right\} &\text{if~~ } (A_{ij}, B_{ij}, C_{ij}, E_{ij}, F_{ij})=(1, 1, 0, 1, 0),\\
    \left\{(0, 0, 1, 0)\right\} &\text{if~~ } (A_{ij}, B_{ij}, C_{ij}, E_{ij}, F_{ij})=(1, 0, 1, 0, 1),\\
    \left\{(0, 1, 0, 1)\right\} &\text{if~~ } (A_{ij}, B_{ij}, C_{ij}, E_{ij}, F_{ij})=(1, 1, 1, 1, 0)\\
    \end{cases}$ for all $i$ and $j$.
    Otherwise, there are no cyclic codes of  length $n$ over $\Z$ whose hulls equal $D$.
\end{thm}
\begin{proof}
    Let $C$ be a cyclic code of odd length $n$ over $\Z$ generated by $\langle f(x)g(x), 2f(x)\rangle$ where $f(x)$ and $g(x)$ are in Eqs \eqref{eq5} and \eqref{eq6}. By Theorem \ref{genhull}, $\Hull(C)$  is generated by
    \begin{align} 
    &\notag\left\langle\lcm\left(f(x)g(x), h^*(x)g^*(x)\right), 2\lcm\left(f(x),  h^*(x)\right)\right\rangle=\\
    &\notag\left\langle\prod_{j| n, j\in N_2}\prod_{i=1}^{\gam(j)}g_{ij}(x)^{\max\{u_{ij}+b_{ij}, 1-u_{ij}\}}\prod_{j| n, j\not\in N_2}\prod_{i=1}^{\be(j)}f_{ij}(x)^{\max\{v_{ij}+z_{ij}, 1-w_{ij}\}}f_{ij}^*(x)^{\max\{w_{ij}+d_{ij}, 1-v_{ij}\}},   \right.\\
    &\label{eq11}\left. 2\prod_{j| n, j\in N_2}\prod_{i=1}^{\gam(j)}g_{ij}(x)^{\max\{u_{ij}, 1-u_{ij}-b_{ij}\}}\prod_{j| n, j\not\in N_2}\prod_{i=1}^{\be(j)}f_{ij}(x)^{\max\{v_{ij}, 1-w_{ij}-d_{ij}\}}f_{ij}^*(x)^{\max\{w_{ij}, 1-v_{ij}-z_{ij}\}} \right\rangle,
    \end{align}
    where $(u_{ij}, b_{ij}), (v_{ij}, z_{ij}), (w_{ij}, d_{ij})\in\left\{(0, 0), (1, 0), (0, 1)\right\}$. 
    
    Comparing the coefficients in Eqs \eqref{eq10} and \eqref{eq11}, we have $A_{ij}=\max\{(u_{ij}+b_{ij}, 1-w_{ij})\}=1$  and $D_{ij}=\max\{u_{ij}, 1-u_{ij}-b_{ij}\}$. Thus $(u_{ij}, b_{ij})=(0, 1)$ if $D_{ij}=0$ and $(u_{ij}, b_{ij})\in \left\{(0, 0), (1, 0)\right\}$ if $D_{ij}=1$.
    
    From Eqs \eqref{eq10} and \eqref{eq11}, it can be deduced that  $B_{ij}=\max\{v_{ij}+z_{ij}, 1-w_{ij}\},$ $ C_{ij}=\max\{w_{ij}+d_{ij}, 1-v_{ij}\}, $ $E_{ij}=\max\{v_{ij}, 1-w_{ij}-d_{ij}\}$ and $F_{ij}=\max\{w_{ij}, 1-v_{ij}-z_{ij}\}$. All $9$  values of $(v_{ij}, z_{ij}, w_{ij}, d_{ij})$ are illustrated in the following table together with their corresponding  values  $(B_{ij}, C_{ij}, E_{ij}, F_{ij})$. 
    \begin{center}{
            \begin{tabular}{|c|c|}
                \hline 
                $(B_{ij}, C_{ij}, E_{ij}, F_{ij})$ &$(v_{ij}, z_{ij}, w_{ij}, d_{ij})$\\\hline
                $(1, 1, 1, 1)$ &$(0, 0, 0, 0), (1, 0, 1, 0)$\\\hline
                $(1, 1, 0, 1)$ & $(0, 1, 1, 0), (0, 0, 0, 1)$\\\hline
                $(1, 1, 1, 0)$ & $(0, 1, 0, 0), (1, 0, 0, 1)$\\\hline
                $(1, 0, 1, 0)$ & $(1, 0, 0, 0)$\\\hline
                $(0, 1, 0, 1)$ & $(0, 0, 1, 0)$\\\hline
                $(1, 1, 0, 0)$ & $(0, 1, 0, 1)$\\\hline
        \end{tabular} }
    \end{center}
    Since there are only  $9$ possible   values of $(v_{ij}, z_{ij}, w_{ij}, d_{ij})$, there are no cyclic codes of length $n$ whose hulls equal $D$ for the other values of $(A_{ij}, B_{ij}, C_{ij}, D_{ij},  E_{ij}, F_{ij})$.
    
    Therefore,  the  proof is completed.
\end{proof}
\subsection{Types and $2$-Dimensions of Hulls of Cyclic Codes}
In this subsection, the types of the hulls of cyclic codes of arbitrary  odd length $n$ over $\Z$ are investigated. Moreover, an  algorithm for finding the types of the hulls of cyclic codes of  odd length $n$ over $\Z$ is given. Finally, the $2$-dimensions of the hulls of cyclic codes of arbitrary odd length $n$ over $\Z$ are determined.

The types of the hulls of cyclic codes of  odd length over $\Z$ is derived in Theorem \ref{thm3.2}. The following lemma is required in its proof.

\begin{lem}\label{3.1}
    Let  $\beta$ be a positive integer. For $1\leq i\leq \beta$, let $(v_{i}, z_{i}), (w_{i}, d_{i})$ and $(u_i, b_i)$ be elements in $\left\{(0, 0), (1, 0), (0, 1)\right\}$. Let $a_{i}=\min\{1-v_{i}-z_{i}, w_{i}\}+\min\{1-w_{i}-d_{i}, v_{i}\}$. Then $a_{i}\in\{0,1\}$. Moreover, the following statements hold.
    \begin{enumerate}
        \item $2-\min\{1-v_i-z_i, w_i\}-\max\{v_i, 1-w_i-d_i\}-\min\{1-w_i-d_i, v_i\}-\max\{w_i, 1-v_i-z_i\}=z_i+d_i$.
        \item If $a_{i}=0$, then $z_{i}+d_{i}\in\{0, 1, 2\}$.
        \item If $a_{i}=1$, then $z_{i}+d_{i}=0$.
        \item Let $a=\sum_{i=1}^{\be}a_{i}$, then $\sum_{i=1}^{\be}(z_{i}+d_{i})=c$ for some $0\leq c\leq 2\left(\be-a\right)$.
    \end{enumerate}
\end{lem}
\begin{proof}
    To prove Statement 1, let $*=1-v_i-z_i, \circ=1-w_i-d_i$ and  $\hexagon=2-\min\{*, w_i\}-\max\{v_i, \circ\}-\min\{v_i, \circ\}-\max\{*, w_i\}$.   Consider the following  table, it can be  concluded that \[2-\min\{*, w_i\}-\max\{*, w_i\}-\min\{v_i, \circ\}-\max\{v_i, \circ\}=\hexagon=z_i+d_i.\]
    \begin{center}
        \begin{tabular}{|c|c|c|c|c|c||c||c|}\hline
            $v_i$&$z_i$&$*$&$w_i$&$d_i$&$\circ$&$z_i+d_i$& $\hexagon$ \\\hline
            $0$&\cellcolor{LightCyan}$0$&$1$&$0$&\cellcolor{LightCyan}$0$&$1$&$0$&$0$\\\hline
            $1$&\cellcolor{LightCyan}$0$&$0$&$0$&\cellcolor{LightCyan}$0$&$1$&$0$&$0$\\\hline
            $0$&\cellcolor{LightCyan}$1$&$0$&$0$&\cellcolor{LightCyan}$0$&$1$&$1$&$1$\\\hline
            $0$&\cellcolor{LightCyan}$0$&$1$&$1$&\cellcolor{LightCyan}$0$&$0$&$0$&$0$\\\hline
            $1$&\cellcolor{LightCyan}$0$&$0$&$1$&\cellcolor{LightCyan}$0$&$0$&$0$&$0$\\\hline
            $0$&\cellcolor{LightCyan}$1$&$0$&$1$&\cellcolor{LightCyan}$0$&$0$&$1$&$1$\\\hline
            $0$&\cellcolor{LightCyan}$0$&$1$&$0$&\cellcolor{LightCyan}$1$&$0$&$1$&$1$\\\hline
            $1$&\cellcolor{LightCyan}$0$&$0$&$0$&\cellcolor{LightCyan}$1$&$0$&$1$&$1$\\\hline
            $0$&\cellcolor{LightCyan}$1$&$0$&$0$&\cellcolor{LightCyan}$1$&$0$&$2$&$2$\\\hline
        \end{tabular}	
    \end{center}

    Clearly,  $a_{i}$ can be either $0$ or $1$.
    Statements 2 and 3  can be easily obtained by considering the possible cases. 
    
    To prove Statement 4,  assume that  $a=\sum_{i=1}^{\be}a_{i}$. Without loss of generality, we assume that $a_{i}=1$ for all $i=1,\ldots, a$. Thus $z_{i}+d_{i}=0$ for all $i=1,\ldots, a$ and $$\sum_{i=1}^{\be}(z_{i}+d_{i})=\left(\sum_{i=1}^{a}z_{i}+d_{i}\right)+\left(\sum_{i=a+1}^{\be}z_{i}+d_{i}\right)=0+\left(\sum_{i=a+1}^{\be}z_{i}+d_{i}\right)=c,$$ where $0\leq c\leq 2\left(\be-a\right)$.
\end{proof}

\begin{thm}\label{thm3.2}
    Let $n$ be an odd positive integer. Then the types of the hull of a cyclic code  of length $n$ over $\Z$ are of the form $4^{k_1}2^{k_2}$, where  $$k_1=\sum_{j| n, j\not\in N_2}\ord_j(2)\cdot a_{j} ~~~~\text{and}~~~~ k_2=\sum_{j| n, j\in N_2} \ord_j(2)\cdot b_j+\sum_{j| n, j\not\in N_2}\ord_j(2)\cdot c_{j},$$
    $0\leq a_j\leq \be(j)$, $0\leq b_j\leq \gam(j)$ and $0\leq c_j\leq 2\left(\be(j)-a_j\right)$.  
\end{thm}
\begin{proof}
    Let $C$ be a cyclic code of length $n$ over $\Z$ generated by  $\langle f(x)g(x), 2f(x)\rangle$,   where $f(x)$, $g(x)$ and  $h(x) $ are monic polynomials such that $x^n-1=f(x)g(x)h(x)$. By Theorem \ref{genhull}, $\Hull(C)$ has type $4^{\deg H(x)}2^{{{\deg G(x)}}}$, where $H(x)$ and $G(x)$  are defined as in Theorem \ref{genhull}. By  Eq \eqref{xn-1}, we have
    \begin{align}
    \label{eq5} f(x)=&\prod_{j| n, j\in N_2}\prod_{i=1}^{\gam(j)}g_{ij}(x)^{u_{ij}}\prod_{j| n, j\not\in N_2}\prod_{i=1}^{\be(j)}f_{ij}(x)^{v_{ij}}f_{ij}^*(x)^{w_{ij}},\\
    \label{eq6} g(x)=&\prod_{j| n, j\in N_2}\prod_{i=1}^{\gam(j)}g_{ij}(x)^{b_{ij}}\prod_{j| n, j\not\in N_2}\prod_{i=1}^{\be(j)}f_{ij}(x)^{z_{ij}}f_{ij}^*(x)^{d_{ij}},\\
    \label{eq7}h(x)=&\prod_{j| n, j\in N_2}\prod_{i=1}^{\gam(j)}g_{ij}(x)^{1-u_{ij}-b_{ij}}\prod_{j| n, j\not\in N_2}\prod_{i=1}^{\be(j)}f_{ij}(x)^{1-v_{ij}-z_{ij}}f_{ij}^*(x)^{1-w_{ij}-d_{ij}},\\ \notag
    f^*(x)=&\prod_{j| n, j\in N_2}\prod_{i=1}^{\gam(j)}g_{ij}(x)^{u_{ij}}\prod_{j| n, j\not\in N_2}\prod_{i=1}^{\be(j)}f_{ij}(x)^{w_{ij}}f_{ij}^*(x)^{v_{ij}},\\\notag
    g^*(x)=&\prod_{j| n, j\in N_2}\prod_{i=1}^{\gam(j)}g_{ij}(x)^{b_{ij}}\prod_{j| n, j\not\in N_2}\prod_{i=1}^{\be(j)}f_{ij}(x)^{d_{ij}}f_{ij}^*(x)^{z_{ij}},\\\notag
    h^*(x)=&\prod_{j| n, j\in N_2}\prod_{i=1}^{\gam(j)}g_{ij}(x)^{1-u_{ij}-b_{ij}}\prod_{j| n, j\not\in N_2}\prod_{i=1}^{\be(j)}f_{ij}(x)^{1-w_{ij}-d_{ij}}f_{ij}^*(x)^{1-v_{ij}-z_{ij}},
    \end{align}	 
    where $(u_{ij}, b_{ij}), (v_{ij}, z_{ij}), (w_{ij}, d_{ij})\in\left\{(0,0), (1, 0), (0, 1)\right\}$. 
    
    First we  determine  $\deg H(x)$.  Observe that 
    \begin{align*}
    H(x)=&\gcd(h(x), f^*(x))\\
    =&\prod_{j| n, j\in N_2}\prod_{i=1}^{\gam(j)}g_{ij}(x)^{\min\{1-u_{ij}-b_{ij}, u_{ij}\}}\\
    &\times \prod_{j| n, j\not\in N_2}\prod_{i=1}^{\be(j)}f_{ij}(x)^{\min\{1-v_{ij}-z_{ij}, w_{ij}\}}f_{ij}^*(x)^{\min\{1-w_{ij}-d_{ij}, v_{ij}\}}\\
    =&\prod_{j| n, j\not\in N_2}\prod_{i=1}^{\be(j)}f_{ij}(x)^{\min\{1-v_{ij}-z_{ij}, w_{ij}\}}f_{ij}^*(x)^{\min\{1-w_{ij}-d_{ij}, v_{ij}\}}
    \end{align*}
    which implies that 
    \begin{align}\notag
    \deg &H(x)=\deg\gcd(h(x), f^*(x)) \notag \\
    =&\deg\prod_{j| n, j\not\in N_2}\prod_{i=1}^{\be(j)}f_{ij}(x)^{\min\{1-v_{ij}-z_{ij}, w_{ij}\}}f_{ij}^*(x)^{\min\{1-w_{ij}-d_{ij}, v_{ij}\}}  \notag \\ \label{degH}
    =&\sum_{j| n, j\not\in N_2}\ord_j(2)\sum_{i=1}^{\be(j)}\left(\min\{1-v_{ij}-z_{ij}, w_{ij}\}+\min\{1-w_{ij}-d_{ij}, v_{ij}\}\right)\\\notag
    =&\sum_{j| n, j\not\in N_2}\ord_j(2)\sum_{i=1}^{\be(j)}a_{ij},~~ \text{where}~0\leq a_{ij}\leq 1.\\\notag
    =&\sum_{j| n, j\not\in N_2}\ord_j(2)\cdot a_{j},~~ \text{where}~0\leq a_{j}\leq \be(j).
    \end{align}
    
    Next we compute $\deg G(x)$. Since
    \begin{align*}
    \lcm(f(x), h^*(x))=
    &\prod_{j| n, j\in N_2}\prod_{i=1}^{\gam(j)}g_{ij}(x)^{\max\{u_{ij}, 1-u_{ij}-b_{ij}\}}\\
    &\times \prod_{j| n, j\not\in N_2}\prod_{i=1}^{\be(j)}f_{ij}(x)^{\max\{v_{ij}, 1-w_{ij}-d_{ij}\}}f_{ij}^*(x)^{\max\{w_{ij}, 1-v_{ij}-z_{ij}\}}
    \end{align*}
    and
    
    \begin{align*} 
    \gcd(h(x), f^*(x))&\cdot\lcm(f(x), h^*(x))=
    \prod_{j| n, j\in N_2}\prod_{i=1}^{\gam(j)}g_{ij}(x)^{\max\{u_{ij}, 1-u_{ij}-b_{ij}\}}\\
    &\times \prod_{j| n, j\not\in N_2}\prod_{i=1}^{\be(j)}f_{ij}(x)^{\min\{1-v_{ij}-z_{ij}, w_{ij}\}+\max\{v_{ij}, 1-w_{ij}-d_{ij}\} }f_{ij}^*(x)^{\min\{1-w_{ij}-d_{ij}, v_{ij}\}+\max\{w_{ij}, 1-v_{ij}-z_{ij}\}},
    \end{align*}
    it  can be deduced that 
    \begin{align*}
    G(x)=&\frac{x^n-1}{\gcd(h(x), f^*(x))\cdot \lcm(f(x), h^*(x))}\\
    =&\prod_{j| n, j\in N_2}\prod_{i=1}^{\gam(j)}g_{ij}(x)^{1-\max\{u_{ij}, 1-u_{ij}-b_{ij}\}}\\
    &\times \prod_{j| n, j\not\in N_2}\prod_{i=1}^{\be(j)}f_{ij}(x)^{1-\min\{1-v_{ij}-z_{ij}, w_{ij}\}-\max\{v_{ij}, 1-w_{ij}-d_{ij}\} }f_{ij}^*(x)^{1-\min\{1-w_{ij}-d_{ij}, v_{ij}\}-\max\{w_{ij}, 1-v_{ij}-z_{ij}\}}.
    \end{align*}
    By Lemma \ref{3.1}, we can conclude that 
    \begin{align}\notag
    \deg G(x)=&\deg\prod_{j| n, j\in N_2}\prod_{i=1}^{\gam(j)}g_{ij}(x)^{1-\max\{u_{ij}, 1-u_{ij}-b_{ij}\}}\\
    \notag&\times \prod_{j| n, j\not\in N_2}\prod_{i=1}^{\be(j)}f_{ij}(x)^{1-\min\{1-v_{ij}-z_{ij}, w_{ij}\}-\max\{v_{ij}, 1-w_{ij}-d_{ij}\} }f_{ij}^*(x)^{1-\min\{1-w_{ij}-d_{ij}, v_{ij}\}-\max\{w_{ij}, 1-v_{ij}-z_{ij}\}}\\
    \notag=&\sum_{j| n, j\in N_2}\ord_j(2)\sum_{i=1}^{\gam(j)}{\left(1-\max\{u_{ij}, 1-u_{ij}-b_{ij}\}\right)}\\
    \notag&+\sum_{j| n, j\not\in N_2}\ord_j(2)\sum_{i=1}^{\be(j)}\left(2-\min\{1-v_{ij}-z_{ij}, w_{ij}\}-\max\{w_{ij}, 1-v_{ij}-z_{ij}\}\right.\\
    &\left.-\min\{1-w_{ij}-d_{ij}, v_{ij}\}-\max\{v_{ij}, 1-w_{ij}-d_{ij}\}\right)\label{degG}\\
    \notag=&\sum_{j| n, j\in N_2}\ord_j(2)\sum_{i=1}^{\gam(j)}\left(1-\max\{u_{ij}, 1-u_{ij}-b_{ij}\}\right)+\sum_{j| n, j\not\in N_2}\ord_j(2)\sum_{i=1}^{\be(j)}(z_{ij}+d_{ij})\notag \\
    \notag=&\sum_{j| n, j\in N_2}\ord_j(2)\cdot b_j+\sum_{j| n, j\not\in N_2}\ord_j(2)\cdot c_{j},
    \end{align}
    
    where $0\leq b_j\leq \gam(j), 0\leq c_{j}\leq 2\left(\be(j)-a_j\right)$.
\end{proof}
The next corollary is an immediate consequence of Theorem \ref{thm3.2} when $n\in N_2$.
\begin{cor}
    Let $n$ be an odd integer such that $n\in N_2$. Then the types of the hull of a cyclic code of length $n$ over $\Z$ are of the form $4^02^{k_2}$, where
    \begin{center}
        $k_2=\displaystyle\sum_{j| n, j\in N_2} \ord_j(2)\cdot b_j$, ~~$0\leq b_j\leq \gam(j)$.
    \end{center}
\end{cor}
\begin{proof}
    Since $n\in N_2$, we have $j\in N_2$ for all $j| n$. Hence, the result follows.
\end{proof}

An algorithm for computing the types of the hull of a cyclic code of odd length $n$ over $\Z$ is given as follows.

\bigskip 
\hrule \smallskip
\centerline{\textbf{Algorithm}: The types of the hull of a cyclic code of odd length $n$ over $\Z$}
\hrule
\begin{enumerate}
    \item For each $j| n$, consider the following cases. 
    \begin{enumerate}
        \item If $j\in N_2$, then compute $\ord_j(2)$ and $\gam(j)$.
        \item If $j\not\in N_2$, then  compute $\ord_j(2)$ and $\be(j)$.
    \end{enumerate}
    \item Compute $k_1=\displaystyle\sum_{j| n, j\not\in N_2}\ord_j(2)\cdot a_j$,  where $0\leq a_j\leq \be(j)$.
    \item For a fixed $a_j$ in 2,  compute $$k_2=\sum_{j| n, j\in N_2}\ord_j(2)\cdot b_j+\sum_{j| n, j\not\in N_2}\ord_j(2)\cdot c_{j},$$ {where}~$0\leq b_j\leq \gam(j)$ and $0\leq c_{j}\leq 2\left(\be(j)-a_j\right)$.
\end{enumerate}
\hrule

\bigskip 

Illustrative examples of  the above algorithm are given as follows. 
\begin{ex}
    Let $n=7$. The types of the hulls of cyclic codes of length $7$ over $\mathbb{Z}_4$ are determined in the following steps.
    \begin{enumerate}
        \item The divisors of $7$ are $1$ and $7$.
        \begin{enumerate}
            \item  $1\in N_2$. We have  $\ord_1(2)=1$ and $\gam(1)=1$.
            \item $7\not\in N_2$. We have  $\ord_7(2)=3$ and $\be(7)=1$.
        \end{enumerate}
        \item Thus $k_1=3a_7$ where $0\leq a_7\leq 1$.
        
        For $a_7=0$, we have  $k_1=0$ and 
        $$k_2=b_1+3c_7,$$
        where $0\leq b_1\leq 1$ and $0\leq c_7\leq 2(1-0)=2$. Thus $k_2\in\{0, 1, 3, 4, 6, 7\}$. Hence,  the types are of the form $4^{k_1}2^{k_2}$,  where $(k_1, k_2)\in\left\{(0, 0), (0, 1), (0, 3),\right.$ $\left.(0, 4), (0, 6), (0, 7)\right\}$.
        
        For $a_7=1$, we have  $k_1=3$ and
        $$k_2=b_1+3c_7=b_1,$$
        where $0\leq b_1\leq 1$ and $0\leq c_7\leq 2(1-1)=0$. Hence, $k_2\in\{0, 1\}$. Therefore the types are of the form $4^{k_1}2^{k_2}$, where $(k_1, k_2)\in\left\{(3, 0), (3, 1)\right\}$. 
    \end{enumerate}
    Altogether, we conclude that the types of the hulls of   cyclic codes of length $7$ over $\mathbb{Z}_4$   are of the form $4^{k_1}2^{k_2}$,  where $(k_1, k_2)\in\left\{(0, 0), (0, 1), (0, 3), (0, 4), (0, 6), (0, 7),\right.$ $\left. (3, 0), (3, 1)\right\}$.
\end{ex}
\begin{ex}
    Let $n=21$. The types of the  hulls of cyclic codes of length $21$ are given as follows.
    \begin{enumerate}
        \item The divisors of $21$ are $1, 3, 7$ and $21$.
        \begin{enumerate}
            \item $1, 3\in N_2$.  We have  $\ord_1(2)=1, \ord_3(2)=2$ and $\gam(1)=1=\gam(3)$. 
            \item $7, 21\not\in N_2$. We have  $\ord_{7}(2)=3, \ord_{21}(2)=6$ and $\be(7)=1=\be(21)$.
        \end{enumerate}
        \item It follows that  $k_1=3a_7+6a_{21}$, where $0\leq a_7, a_{21}\leq 1$.
        
        For $(a_7, a_{21})=(0, 0)$,  we have  $k_1=0$ and 
        $$k_2=b_1+2b_3+3c_7+6c_{21},$$
        where $0\leq b_1, b_3\leq 1$ and $0\leq c_7, c_{21}\leq 2$. So $k_2\in \{0, 1,\ldots, 21\}$.
        
        For $(a_7, a_{21})=(1, 0)$, we have  $k_1=3$ and 
        $$k_2=b_1+2b_3+3c_7+6c_{21},$$
        where $0\leq b_1, b_3 \leq 1$, $c_7=0$ and $0\leq c_{21}\leq 2$. Hence, $k_2\in\{0, 1, 2, 3, 6, 7, 8,$ $9, 12, 13, 14, 15\}$.
        
        For $(a_7, a_{21})=(0, 1)$,  we have  $k_1=6$ and 
        $$k_2=b_1+2b_3+3c_7+6c_{21},$$
        where $0\leq b_1, b_3 \leq 1$, $0\leq c_7\leq 2$ and $c_{21}=0$. Thus $k_2\in\{0, 1,\ldots, 9\}$.

        For $(a_7, a_{21})=(1, 1)$, then $k_1=9$ and 
        $$k_2=b_1+2b_3+3c_7+6c_{21},$$
        where $0\leq b_1, b_3 \leq 1$, $c_7=0$ and $c_{21}=0$. Hence, $k_2\in\{0, 1, 2, 3\}$.
    \end{enumerate}
\end{ex}

For each odd integer $3\leq n \leq 35$,  the types  $4^{k_1}2^{k_2}$ of the hulls of cyclic codes of length $n$ over $\Z$ are given in Table \ref{tab} based on the above algorithm.

\begin{table}[!hbt]
    \centering
    \caption{The types $4^{k_1}2^{k_2}$ of the hulls of cyclic codes of odd lengths up to $35$}\label{tab}
    \begin{tabular}{|c|c|c|}
        \hline 
        $n$ & $k_1$ & $k_2$  \\ 
        \hline 
        $3$& $0$ &  $0, 1, 2, 3$\\ 
        \hline 
        $5$ &$0$ &$0, 1, 4, 5$  \\ 
        \hline 
        $7$& $0$ & $0, 1, 3, 4, 6, 7$ \\ 
        \hline 
        & $3$ &  $0, 1$\\ 
        \hline 
        $9$	&$0$  & $0, 1, 2, 3, 6, 7, 8, 9$  \\ 
        \hline 
        $11$	& $0$ & $0, 1, 10, 11$ \\ 
        \hline 
        $13$& $0$ & $0, 1, 12, 13$ \\ 
        \hline 
        $15$& $0$  & $0, 1,\ldots, 15$ \\ 
        \hline 
        &$4$  &  $0, 1, 2, 3, 4, 5, 6, 7$\\ 
        \hline 
        $17$& $0$ & $0, 1, 8, 9, 16, 17$\\\hline 
        $19$&$0$&$0, 1, 18, 19$\\\hline 
        $21$&$0$&$0, 1, 2,\ldots, 21$\\\hline 
        &$3$&$0, 1, 2, 3, 6, 7, 8, $\\\hline 
        &&$9, 12, 13, 14, 15$\\\hline
        &$6$&$0, 1,\ldots, 9$\\\hline 
        &$9$&$0, 1, 2, 3$\\\hline 
        $23$&$0$&$0, 1, 11, 12, 22, 23$\\\hline
        &$11$&$0, 1$\\\hline
        $25$&$0$&$0, 1, 4, 5, 20, 21, 24, 25$\\\hline
    \end{tabular} \hspace{0.5cm}
    \begin{tabular}{|c|c|c|}\hline
        $n$&$k_1$&$k_1$\\\hline	
        $27$&$0$&$0, 1, 2, 3, 6, 7, 8, 9, 18,$\\\hline
        &&$19, 20, 21, 24, 25, 26, 27$\\\hline
        $29$&$0$&$0, 1, 28, 29$\\\hline
        $31$&$0$&$0, 1, 5, 6, 10, 11, 15, 16,$\\\hline
        &&$20, 21, 25, 26, 30, 31$\\\hline
        &$5$&$0, 1, 5, 6, 10, 11, 15, 16$\\\hline
        &&$20, 21$\\\hline
        &$10$&$0, 1, 5, 6, 10, 11$\\\hline
        &$15$&$0, 1$\\\hline
        $33$&$0$&$0, 1, 2, 3, 10, 11, 12, 13, 20$ \\\hline
        &&$21, 22, 23, 30, 31, 32, 33$\\\hline
        $35$&$0$&$0, 1, 3, 4, 5, 6, 7, 8, 10, $\\\hline
        &&$11, 12, 13, 15, 16, 17, 18,   $\\\hline
        &&$19, 20, 22, 23, 24, 25, 27,$\\\hline
        &&$ 28, 29, 30, 31, 32, 34, 35$\\\hline
        &$3$&$0, 1, 4, 5, 12, 13, 16, 17,$\\\hline
        &&$ 24, 25, 28, 29$\\\hline
        &$12$&$0, 1, 3, 4, 5, 6, 7, 8, 10, 11$\\\hline
        &$15$&$0, 1, 4, 5$\\\hline
    \end{tabular}
\end{table}

A formula for the $2$-dimensions of the hulls of cyclic codes of odd length $n$ over $\Z$ is given as follows.
\begin{thm} \label{thm3.10}
    Let $n$ be an odd positive integer. Then the $2$-dimensions of the hull of cyclic codes of length $n$ over $\Z$ are of the form
    \begin{align}\label{2-dim}
    \sum_{j\mid n, j\in N_2}\ord_j(2)\cdot\triangle_j+\sum_{j\mid n, j\not\in N_2}\ord_j(2)\cdot\blacktriangle_j,
    \end{align}
    where $0\leq \triangle_j\leq\gamma(j)$ and $0\leq \blacktriangle_j\leq 2\be(j)$.
\end{thm}
\begin{proof}
    Let  $C$ be a cyclic code of odd length $n$ over $\Z$   generated by $\langle f(x)g(x), 2f(x)\rangle$, where $x^n-1=f(x)g(x)h(x)$. By Theorem \ref{genhull},  we have that  $\Hull(C)$ has type $4^{\deg H(x)}2^{{{\deg G(x)}}}$ and the $2$-dimension of $\Hull(C)$ is \[2\deg H(x)+\deg G(x),\] where   \[H(x)={\gcd(h(x), f^*(x))} \]  and \[ G(x)={{\frac{x^n-1}{\gcd(h(x), f^*(x))\cdot\lcm(f(x), h^*(x))}}}.\]
    By Eqs \eqref{degH} and \eqref{degG}, it can be deduced that 
    \begin{align}
    \notag\dim_2&\left(\Hull (C)\right)=2\deg H(x)+\deg G(x)\\
    \notag=&~2\sum_{j| n, j\not\in N_2}\ord_j(2)\sum_{i=1}^{\be(j)}\left(\min\{1-v_{ij}-z_{ij}, w_{ij}\}+\min\{1-w_{ij}-d_{ij}, v_{ij}\}\right)	\\
    \notag	&+\sum_{j| n, j\in N_2}\ord_j(2)\sum_{i=1}^{\gam(j)}{\left(1-\max\{u_{ij}, 1-u_{ij}-b_{ij}\}\right)}\\
    \notag	&+\sum_{j| n, j\not\in N_2}\ord_j(2)\sum_{i=1}^{\be(j)}\left(2-\min\{1-v_{ij}-z_{ij}, w_{ij}\}-\max\{v_{ij}, 1-w_{ij}-d_{ij}\}\right.\\ \notag&\left.-\min\{1-w_{ij}-d_{ij}, v_{ij}\}-\max\{w_{ij}, 1-v_{ij}-z_{ij}\}\right)\\
    \notag	=&\sum_{j| n, j\in N_2}\ord_j(2)\sum_{i=1}^{\gam(j)}{\left(1-\max\{u_{ij}, 1-u_{ij}-b_{ij}\}\right)}\\
    \notag	&+\sum_{j| n, j\not\in N_2}\ord_j(2)\sum_{i=1}^{\be(j)}\left(2+\min\{1-v_{ij}-z_{ij}, w_{ij}\}-\max\{v_{ij}, 1-w_{ij}-d_{ij}\}\right.\\ \label{eq13}
    &\left.+\min\{1-w_{ij}-d_{ij}, v_{ij}\}-\max\{w_{ij}, 1-v_{ij}-z_{ij}\}\right)\\
    \label{eq14}=& \sum_{j| n, j\in N_2}\ord_j(2)\sum_{i=1}^{\gam(j)}\triangle_{ij}+\sum_{j| n, j\not\in N_2}\ord_j(2)\sum_{i=1}^{\be(j)}\blacktriangle_{ij},
    \end{align}
    where $\triangle_{ij}=1-\max\{u_{ij}, 1-u_{ij}-b_{ij}\}$ and 
    \begin{align*}
    \blacktriangle_{ij}=2&+\min\{1-v_{ij}-z_{ij}, w_{ij}\}-\max\{v_{ij}, 1-w_{ij}-d_{ij}\}\\
    &+\min\{1-w_{ij}-d_{ij}, v_{ij}\}-\max\{w_{ij}, 1-v_{ij}-z_{ij}\}.
    \end{align*}
    It is not difficult to see that $0\leq\triangle_{ij}\leq 1$ and $0\leq\blacktriangle_{ij}\leq 2$. From Eq \eqref{eq14}, we have
    \begin{align*}
    \dim_2(\Hull(C))
    &= \sum_{j| n, j\in N_2}\ord_j(2)\sum_{i=1}^{\gam(j)}\triangle_{ij}+\sum_{j| n, j\not\in N_2}\ord_j(2)\sum_{i=1}^{\be(j)}\blacktriangle_{ij}\\&=\sum_{j| n, j\in N_2}\ord_j(2)\cdot\triangle_{j}+\sum_{j| n, j\not\in N_2}\ord_j(2)\cdot\blacktriangle_{j},
    \end{align*}
    where $\triangle_j=\sum_{i=1}^{\gam(j)}\triangle_{ij}$ and $\blacktriangle_{j}=\sum_{i=1}^{\be(j)}\blacktriangle_{ij}$. It follows  that $0\leq \triangle_j\leq\gam(j)$ and $0\leq \blacktriangle_{j}\leq2\be(j)$. Hence, the $2$-dimension of $C$  is of the form in  Eq \eqref{2-dim}.
\end{proof}
\subsection{Enumeration of Cyclic Codes of  the same $2$-Dimension}

The $2$-dimensions of the hulls of cyclic codes of odd length $n$ over $\Z$ are determined in the previous subsection. Here, the number of cyclic codes of odd length $n$ over $\Z$ whose  hulls share  a fixed $2$-dimension is investigated. 

Let $\ell$ denote a $2$-dimension of the hull of cyclic codes of odd length $n$ over $\Z$ given in Eq \eqref{2-dim}. Later, the number of cyclic codes of odd length $n$ over $\Z$  whose  hulls have  $2$-dimension $\ell$ will  be obtained in terms of the solutions $\triangle_{ij}$'s and $\blacktriangle_{ij}$'s of
$$\sum_{j\mid n, j\in N_2}\ord_j(2)\sum_{i=1}^{\gam(j)}\triangle_{ij}+\sum_{j\mid n, j\not\in N_2}\ord_j(2)\sum_{i=1}^{\be(j)}\blacktriangle_{ij},$$
where $0\leq \triangle_{ij}\leq 1$ and $0\leq\blacktriangle_{ij}\leq 2$.  For convenience, let $((\triangle_{ij}))$ be a vector whose entries are $0\leq \triangle_{ij}\leq 1$ and the indices satisfy $j\mid n, j\in N_2$ and $1\leq i\leq\gam(j)$, i.e.,
$$((\triangle_{ij})):=(\triangle_{ij})_{j\mid n, j\in N_2, 1\leq i\leq\gam(j)}.$$
Similarly, let
$$((\blacktriangle_{ij})):=(\blacktriangle_{ij})_{j\mid n, j\not\in N_2, 1\leq i\leq\be(j)},$$
where $0\leq \blacktriangle_{ij}\leq 2$. Denote by $(((\triangle_{ij})), ((\blacktriangle_{ij})))$ the concatenation of the vectors $((\triangle_{ij}))$ and $((\blacktriangle_{ij}))$.
\begin{thm} \label{thm3.11}
    Let $n$ be an odd positive integer and $\ell$ be in the form of Eq \eqref{2-dim}. The number of cyclic codes of length $n$ over $\Z$ whose hulls have $2$-dimension $\ell$ is 
    \begin{align*}
    \sum_{(((\triangle_{ij})), ((\blacktriangle_{ij})))\in h(\ell)}\left(\prod_{j\mid n, j\in N_2}\prod_{i=1}^{\gam(j)}(2-\triangle_{ij})\prod_{j\mid n, j\not\in N_2}\prod_{i=1}^{\be(j)}\left(-\frac{3}{2}\blacktriangle_{ij}^2+\frac{7}{2}\blacktriangle_{ij}+2\right)\right),
    \end{align*}
    where $$h(\ell)=\left\{(((\triangle_{ij})), ((\blacktriangle_{ij})))~\middle|~ \sum_{j\mid n, j\in N_2}\ord_j(2)\sum_{i=1}^{\gam(j)}\triangle_{ij}+\sum_{j\mid n, j\not\in N_2}\ord_j(2)\sum_{i=1}^{\be(j)}\blacktriangle_{ij}=\ell \right\}.$$
\end{thm}
\begin{proof}
    For a fixed $(((\triangle_{ij})), ((\blacktriangle_{ij})))$, we want to find the polynomials $f(x), g(x)$ and $h(x)$ in Eqs \eqref{eq5}, \eqref{eq6} and \eqref{eq7} such that the $2$-dimension of the hull of a cyclic code generated by $\langle f(x)g(x), 2f(x)\rangle$ is 
    \begin{align}\label{ell}
    \sum_{j\mid n, j\in N_2}\ord_j(2)\sum_{i=1}^{\gam(j)}\triangle_{ij}+\sum_{j\mid n, j\not\in N_2}\ord_j(2)\sum_{i=1}^{\be(j)}\blacktriangle_{ij}=\ell.
    \end{align}
    By Eqs \eqref{eq14}, it can be deduced that 
    \begin{align*}
    \notag\dim_2\left(\Hull (C)\right)=&~2\deg H(x)+\deg G(x)\\
    &= \sum_{j| n, j\in N_2}\ord_j(2)\sum_{i=1}^{\gam(j)}\triangle_{ij}+\sum_{j| n, j\not\in N_2}\ord_j(2)\sum_{i=1}^{\be(j)}\blacktriangle_{ij},
    \end{align*}
    where $\triangle_{ij}=1-\max\{u_{ij}, 1-u_{ij}-b_{ij}\}$ and 
    \begin{align*}
    \blacktriangle_{ij}=2&+\min\{1-v_{ij}-z_{ij}, w_{ij}\}-\max\{v_{ij}, 1-w_{ij}-d_{ij}\}\\
    &+\min\{1-w_{ij}-d_{ij}, v_{ij}\}-\max\{w_{ij}, 1-v_{ij}-z_{ij}\}.
    \end{align*}
    
    For given  $\triangle_{ij}$ and $\blacktriangle_{ij}$, the values of $(u_{ij}, b_{ij})$ and $(v_{ij}, z_{ij}, w_{ij}, d_{ij})$ are listed respectively in the following tables.
    \begin{table}[h!]\centering
        \begin{tabular}{|c|c|}\hline
            $\triangle_{ij}$ & $(u_{ij}, b_{ij})$ \\\hline
            $0$& $(0, 0), (1, 0)$\\\hline
            $1$& $(0, 1)$\\\hline
        \end{tabular}\quad
        \begin{tabular}{|c|c|}\hline
            $\blacktriangle_{ij}$&$(v_{ij}, z_{ij}, w_{ij}, d_{ij})$\\\hline
            $0$&$(0, 0, 0, 0), (1, 0, 1, 0)$\\\hline
            $1$&$(0, 1, 0, 0), (0, 1, 1, 0), (0, 0, 0, 1), (1, 0, 0, 1)$\\\hline
            $2$&$(1, 0, 0, 0), (0, 0, 1, 0), (0, 1, 0, 1)$ \\\hline
        \end{tabular}
    \end{table}
    
    Thus, for a given $(((\triangle_{ij})), ((\blacktriangle_{ij})))$, the number of cyclic codes of odd length $n$ over $\Z$ whose hulls have  $2$-dimension  in the form of  Eq \eqref{ell} is
    \begin{align}\label{eq17}
    \prod_{j\mid n, j\in N_2}\prod_{i=1}^{\gam(j)}(2-\triangle_{ij})\prod_{j\mid n, j\not\in N_2}\prod_{i=1}^{\be(j)}\left(-\frac{3}{2}\blacktriangle_{ij}^2+\frac{7}{2}\blacktriangle_{ij}+2\right).
    \end{align}
    Therefore, the number of cyclic codes of odd length $n$ over $\Z$ having hulls of $2$-dimension $\ell$ is the summations of \eqref{eq17} where all $(((\triangle_{ij})), ((\blacktriangle_{ij})))$ runs in the set $h(\ell)$.
\end{proof}

We note that  the  average $2$-dimension $E(n)$ of the hull of cyclic codes of length $n$ over $\Z$   can be given in terms of the fraction of the sum of the number of cyclic codes whose hulls have $2$-dimension $\ell$ in Theorem \ref{thm3.11}, where $\ell$ runs over all the  $2$-dimensions in Theorem  \ref{thm3.10} and the number of cyclic codes ${|\mathcal{C}(n, 4)|}$. Using this direction,  it might lead to a very  tedious calculation. Here, 
an alternative simpler way to determine  $E(n)$ is given in the next section using  probability theory. 

\section{The Average $2$-Dimension $E(n)$} 

Recall that  $n$ is an odd positive integer, $\mathcal{C}(n, 4)$ is  the set of all cyclic codes of length $n$ over $\Z$ and the  average $2$-dimension of the hull of cyclic codes of length $n$ over $\Z$ is 
\begin{align*} E(n)=\sum_{C\in\mathcal{C}(n, 4)}\frac{\dim_2(\Hull (C))}{|\mathcal{C}(n, 4)|}.
\end{align*}  
In this section, an  explicit formula  of $E(n)$ and its  upper bounds are given in terms of $B_n$ and the length  $n$ of the codes.

First, we prove the following useful    expectations. 
\begin{lem}\label{E}
    Let $(v, z), (w, d), (u, b)\in\left\{(0, 0), (1, 0), (0, 1)\right\}$. Then 
    \begin{enumerate}
        \item $E\left(1-\max\{u, 1-u-b\}\right)=\frac{1}{3}$.
        \item $E\left(2+\min\{1-v-z, w\}-\max\{v, 1-w-d\}+\min\{1-w-d, v\}\right.\\
        \left.-\max\{w, 1-v-z\}\right)=\frac{10}{9}$.
    \end{enumerate}
\end{lem}
\begin{proof}
    To prove 1,  consider the values in the following  table. Let $\triangle=1-\max\{u, 1-u-b\}$.
    
    \begin{table}[h!]\centering
        \begin{tabular}{|c|c|c|c|c|}
            \hline
            $u$&$b$&$1-u-b$&$\max\{u, 1-u-b\}$&$\triangle$\\\hline
            $0$&$0$&$1$&$1$&$0$\\\hline
            $1$&$0$&$0$&$1$&$0$\\\hline
            $0$&$1$&$0$&$0$&$1$\\\hline
        \end{tabular}
    \end{table}
    {It follows that $E\left(\triangle\right)=0\cdot\frac{2}{3}+1\cdot\frac{1}{3}=\frac{1}{3}$.}
    
    To prove 2, let $*=1-v-z, \circ=1-w-d$ and $\blacktriangle=2+\min\{*, w\}-\max\{*, w\}+\min\{v, \circ\}-\max\{v, \circ\}$.  
    
    \begin{center} 
        \begin{tabular}{|c|c|c|c|c|c||c|c|c|c||c|}\hline
            $v$&$z$&$*$&$w$&$d$&$\circ$&$\min\{*, w\}$&$\max\{*, w\}$&$\min\{v, \circ\}$&$\max\{v, \circ\}$&$\blacktriangle$\\\hline
            $0$&$0$&$1$&$0$&$0$&$1$&$0$&$1$&$0$&$1$&$0$\\\hline
            $1$&$0$&$0$&$0$&$0$&$1$&$0$&$0$&$1$&$1$&$2$\\\hline
            $0$&$1$&$0$&$0$&$0$&$1$&$0$&$0$&$0$&$1$&$1$\\\hline
            $0$&$0$&$1$&$1$&$0$&$0$&$1$&$1$&$0$&$0$&$2$\\\hline
            $1$&$0$&$0$&$1$&$0$&$0$&$0$&$1$&$0$&$1$&$0$\\\hline
            $0$&$1$&$0$&$1$&$0$&$0$&$0$&$1$&$0$&$0$&$1$\\\hline
            $0$&$0$&$1$&$0$&$1$&$0$&$0$&$1$&$0$&$0$&$1$\\\hline
            $1$&$0$&$0$&$0$&$1$&$0$&$0$&$0$&$0$&$1$&$1$\\\hline
            $0$&$1$&$0$&$0$&$1$&$0$&$0$&$0$&$0$&$0$&$2$\\\hline
        \end{tabular}
    \end{center}
    
    From the above  table, it can be concluded that   \begin{align*}
    E(\blacktriangle)&=E\left(2+\min\{*, w\}-\min\{*, w\}+\min\{v, \circ\}-\max\{v, \circ\}\right)\\&=0\cdot\frac{2}{9}+1\cdot\frac{4}{9}+2\cdot\frac{3}{9}=\frac{10}{9}.
    \end{align*}
    The proof is completed. 
\end{proof}
From Eqs \eqref{eq5}, \eqref{eq6} and \eqref{eq7}, the following lemma can be deduced directly.
\begin{lem}\label{lem4.2}
    There is a bijection between $\mathcal{C}(n, 4)$ and the set \begin{align*}
    S=&\left\{((u_1, b_1),\ldots,(u_\mathtt{s}, b_\mathtt{s}),(v_1, z_1),\ldots, (v_\mathtt{t}, z_\mathtt{t}), (w_1, d_1),\ldots, (w_\mathtt{t}, d_\mathtt{t}))\mid\right.\\
    &\left. (u_i, b_i), (v_j, z_j), (w_j, d_j)\in\left\{(0, 0), (1, 0), (0, 1)\right\} \text{ for all } 0\leq i\leq \mathtt{s} \text{ and }0\leq j\leq \mathtt{t}\right\},  \end{align*}
    where $ \mathtt{s}$ and $\mathtt{t}$ are given in  Eqs \eqref{eqs} and \eqref{eqt}.
\end{lem}

Based on Lemma \ref{E}, the formula for the average $2$-dimension of the hull of cyclic codes of odd length $n$  over $\Z$ can be determined using the expectation $E(Y)$, where  $Y$ is  the random variable  of the $2$-dimension  $\dim_2(\Hull (C))$ and  $C$  is chosen randomly from  $\mathcal{C}(n, 4)$ with uniform probability.  

\begin{thm}\label{11}
    Let $n$ be an odd positive integer. Then the average $2$-dimension of the hull of cyclic codes of length $n$ over $\Z$ is
    \begin{align*}
    E(n)=\frac{5}{9}n-\frac{2}{9}B_n,
    \end{align*}
    where $B_n$ is defined in Eq \eqref{Bn}.
\end{thm}
\begin{proof}Let $C$ be a cyclic code of length $n$ over $\Z$ generated by  \[\langle f(x)g(x), 2f(x)\rangle,\] where $x^n-1=f(x)g(x)h(x) $ and $f(x)$, $g(x)$ and  $h(x) $ are monic polynomials. By Theorem \ref{genhull},  we have $\Hull(C)$ has type $4^{\deg H(x)}2^{{{\deg G(x)}}}$ and the $2$-dimension of $\Hull(C)$ is \[2\deg H(x)+\deg G(x),\] where   \[H(x)={\gcd(h(x), f^*(x))} \]  and \[ G(x)={{\frac{x^n-1}{\gcd(h(x), f^*(x))\cdot\lcm(f(x), h^*(x))}}}.\]

    Let $Y$ be the random variable of the $2$-dimension $\dim_2(C)$, where $C$ is chosen randomly from $\mathcal{C}(n, 4)$ with uniform probability. Let $E(Y)$ be the expectation of $Y$. Thus $E(n)=E(Y)$. Therefore, choosing a cyclic code $C$ from $\mathcal{C}(n, 4)$ with uniform probability $\frac{1}{|\mathcal{C}(n, 4)|}=\frac{1}{3^{\mathtt{s}+2\mathtt{t}}}$ and choosing an element in $S$ defined in Lemma \ref{lem4.2} with uniform probability $\frac{1}{3^{\mathtt{s}+2\mathtt{t}}}$ are identical. By Theorem \ref{genhull}, we obtain
    \begin{align}\label{eq18}
    Y=\dim_2(\Hull(C))=2\deg H(x)+\deg G(x).
    \end{align}
    From Eqs \eqref{eq13} and \eqref{eq18},  we have 
    \begin{align*}
    E(n)=&~E(Y)\\
    =&~E(2\deg H(x)+\deg G(x))\\
    =&~E\left(\sum_{j| n, j\in N_2}\ord_j(2)\sum_{i=1}^{\gam(j)}{\left(1-\max\{u_{ij}, 1-u_{ij}-b_{ij}\}\right)}\right)\\
    &+E\left(\sum_{j| n, j\not\in N_2}\ord_j(2)\sum_{i=1}^{\be(j)}\left(2+\min\{1-v_{ij}-z_{ij}, w_{ij}\}-\max\{v_{ij}, 1-w_{ij}-d_{ij}\}\right.\right.\\
    &\left.\left.+\min\{1-w_{ij}-d_{ij}, v_{ij}\}-\max\{w_{ij}, 1-v_{ij}-z_{ij}\}\right)\right)\\
    =&\sum_{j| n, j\in N_2}\ord_j(2)\cdot{\gam(j)}\cdot{E\left(1-\max\{u_{ij}, 1-u_{ij}-b_{ij}\}\right)}+\sum_{j| n, j\not\in N_2}\ord_j(2)\cdot \be(j)\\
    &\cdot E\left(2+\min\{1-v_{ij}-z_{ij}, w_{ij}\}-\max\{v_{ij}, 1-w_{ij}-d_{ij}\}\right.\\
    &\left. +\min\{1-w_{ij}-d_{ij}, v_{ij}\}-\max\{w_{ij}, 1-v_{ij}-z_{ij}\}\right)\\
    =&\sum_{j| n, j\in N_2}\phi(j)\cdot \frac{1}{3}+\sum_{j| n, j\not\in N_2}\frac{\phi(j)}{2}\cdot\frac{10}{9} ~~\text{by Lemma \ref{E}},\\
    =&~\frac{B_n}{3}+\frac{5(n-B_n)}{9}~~\text{by Eq \eqref{Bn},}\\
    =&~\frac{5n}{9}-\frac{2B_n}{9}.
    \end{align*} 
    The proof is completed.
\end{proof}
The next corollary is a  direct consequence of Theorem \ref{11}.
\begin{cor}
    Assume the notations as in Theorem \ref{11}. Then $E(n)<\frac{5n}{9}$.
\end{cor}

The   average $2$-dimension $E(n)$ of the hull of cyclic codes of odd  length up to $53$ over $\Z$  are given in  Table \ref{tab1}. The row is highlighted in gray when    $n\in N_2$ and  $n\not\in N_2$ otherwise.
\renewcommand{\arraystretch}{1.8}

\begin{table}[!hbt]\centering
    \begin{tabular}{|c|c|c|}
        \hline
        $n$&$B(n)$&$E(n)=\frac{5n-2B_n}{9}$\\
        \hline
        \rowcolor{LightCyan}$3$&$3$&$1$ \\
        \hline
        \rowcolor{LightCyan}$5$&$5$&$\frac{5}{3}$ \\\hline
        $7$&$1$&$\frac{11}{3}$ \\\hline         
        \rowcolor{LightCyan}$9$&$9$&$3$ \\\hline
        \rowcolor{LightCyan}$11$&$11$&$\frac{11}{3}$ \\\hline
        \rowcolor{LightCyan}$13$&$13$&$\frac{13}{3}$\\\hline      
        $15$&$7$&$\frac{61}{9}$ \\\hline
        \rowcolor{LightCyan}$17$&$17$&$\frac{17}{3}$ \\\hline
        \rowcolor{LightCyan}$19$&$19$&$\frac{19}{3}$\\\hline  
        $21$&$3$&$11$ \\\hline
        $23$&$1$&$\frac{113}{9}$ \\\hline
        \rowcolor{LightCyan}$25$&$25$&$\frac{25}{3}$ \\\hline     
        \rowcolor{LightCyan}$27$&$27$&$9$ \\\hline  		       	
    \end{tabular}
    \begin{tabular}{|c|c|c|}
        \hline
        $n$&$B(n)$&$E(n)=\frac{5n-2B_n}{9}$\\\hline
        \rowcolor{LightCyan}$29$&$29$&$\frac{29}{3}$ \\\hline
        $31$&$1$&$17$ \\\hline
        \rowcolor{LightCyan}$33$&$33$&$11$ \\\hline         
        $35$&$5$&$\frac{55}{3}$ \\\hline
        \rowcolor{LightCyan}$37$&$37$&$\frac{37}{3}$ \\\hline
        $39$&$15$&$\frac{55}{3}$ \\\hline      
        \rowcolor{LightCyan}$41$&$41$&$\frac{41}{3}$ \\\hline
        \rowcolor{LightCyan}$43$&$43$&$\frac{43}{3}$ \\\hline
        $45$&$13$&$\frac{199}{9}$ \\\hline  
        $47$&$1$&$\frac{233}{9}$ \\\hline
        $49$&$1$&$27$ \\\hline
        $51$&$19$&$\frac{217}{9}$ \\\hline     
        \rowcolor{LightCyan}$53$&$53$&$\frac{53}{3}$ \\\hline  		       	
    \end{tabular}
    \caption{$E(n)$  of cyclic codes of odd length $n$ up to $53$ over $\Z$.}\label{tab1}
\end{table}

\section{$N_2$-factorization and Bounds on $E(n)$}
In this section, a simplified formula of $B_n$ is derived. Lower and upper bounds for $E(n)$ can be obtained using this  formula of  $B_n$.

Recall that $N_2=\left\{\ell\geq 1 : \ell~ \text{divides}~ 2^i+1 \text{ for some positive integer }i \right\}$.
\begin{lem} Let $\ell$ be a positive integer. 
    If $\ell\in N_2$, then $\ord_\ell(2)$ is even.  
\end{lem}
\begin{proof}
    Assume that $\ell\in N_2$. Then there exists the smallest positive integer $k$ such that $\ell| (2^k+1)$, which implies $\ell| (2^{2k}-1)$. So $\ord_{\ell}(2)| 2k$. Since $\ord_{\ell}(2)\not\mid k$, $\ord_{\ell}(2)$ is even.
\end{proof}
Let $P_\al:=\left\{\ell\in N_2 : 2^\al || \ord_{\ell}(2)\right\}$, where the notation $2^\ell || k$ means that  $\al$ is the non-negative integer such that $2^\al| k$ but  $2^{\al+1}\not\mid k$. Clearly, $P_0=\{1\}$.
\begin{thm}[{\cite[Theorem 4]{Skersys03}}]
    Let $\ell>1$ be an odd positive integer. Let $\ell=p_1^{e_1}\ldots p_k^{e_k}$ be the prime factorization of $\ell$. Then $\ell\in N_2$ if and only if there exists $\al\geq 1$ such that $p_i\in P_\al$ for all $i$. In this case, we have $\ell\in P_\al$.
\end{thm} 
\begin{lem}\label{lem43}
    Let $\al\geq 1$ an integer  and let $\ell$ be a positive integer. If $\ell\in P_{\al}$, then $\ell\geq 2^\al+1$.
\end{lem}
\begin{proof}
    Note that $\ell\geq 3$. Since $\ell\in P_{\al}$, it follows that $ 2^\al||\ord_\ell(q)$. By  Fermat's Little Theorem, we have $\ord_{\ell}(q)|\phi(\ell)$. Then $2^\al|\phi(\ell)$.  Hence,  $2^\al\leq \phi(\ell)\leq \ell-1$.
\end{proof}

Let $\ell=p_1^{e_1}\ldots p_k^{e_k}$ be the prime factorization of $\ell$, where $p_1,\ldots, p_k$ are distinct odd primes and $e_i\geq 1$ for all $1\leq i\leq k$. Partition the index set $\{1,\ldots, k\}$ into $K', K_1, K_2,\ldots$ as follows:
\begin{enumerate}
    \item $K'=\left\{i\mid p_i\not\in N_2\right\}$,
    \item $K_\al=\left\{i\mid p_i\in N_2 \text{~~and~~} p_i\in P_\al\right\}$.
\end{enumerate}
Let $d'=\prod_{i\in K'}p_i^{e_i}$ and $d_\al=\prod_{i\in K_\al}p_i^{e_i}$. For convenience, the empty product will be regarded as $1$. Therefore, we have $\ell=d'd_1d_2\ldots$ which is called the $N_2$\textit{-factorization} of $\ell$.
\begin{lem}[{\cite[Lemma 9]{Skersys03}}]\label{notinN2}
    Let $\ell$ be an odd integer and let $\ell=d'd_1d_2\ldots$ be the $N_2$-factorization of $\ell$. If $\ell\not\in N_2$, then at least one of the following conditions hold.
    \begin{enumerate}
        \item $d'>1$.
        \item $d_{\al_1}>1$ and $d_{\al_2}>1$ for two distinct $\al_1\geq 1$ and $\al_2\geq 1$.
    \end{enumerate}
\end{lem}
\begin{prop}\label{Prop5.5}
    Let $n$ be an odd integer and let $n=d'd_1d_2\ldots$ be an $N_2$-factorization of $n$. If $n\not\in N_2$, then $B_n=d_1+\sum_{\al\geq 2}{(d_\al-1)}$.
\end{prop}
\begin{proof}
    Since $\sum_{i| d}\phi(i)=d$,  we have
    \[B_n=\sum_{j| n, j\in N_2}\phi(j)=\phi(1)+\sum_{\al\geq 1}\sum_{k| d_{\al},~k\neq 1}\phi(k)=1+\sum_{\al\geq 1}(d_\al-1)=d_1+\sum_{\al\geq 2}(d_\al-1)
    \] by Eq \eqref{Bn}.
\end{proof}
Applying Lemma \ref{lem43}, Lemma \ref{notinN2} and Proposition \ref{Prop5.5},  some upper and lower bounds of  $E(n)$ can be concluded in the following theorem.
\begin{thm}\label{final}
    Let $n$ be an odd integer. The the following statements hold.
    \begin{enumerate}
        \item $n\in N_2$ if and only if $E(n)=\frac{n}{3}$.
        \item If $n\not\in N_2$,  then $\frac{11n}{27}\leq E(n)\leq \frac{5n}{9}$.
    \end{enumerate}
\end{thm}
\begin{proof}
    To prove 1, let $n\in N_2$. Then $B_n=\sum_{j| n, j\in N_2}\phi(j)= \sum_{j| n}\phi(j)=n$.

    Conversely, we assume that $E(n)=\frac{n}{3}$. By Theorem \ref{11}, we have $\frac{n}{3}=\frac{5n}{9}-\frac{2B_n}{9}$. Thus $n=B_n$, which implies $n\in N_2$.

    To prove 2, let $n\not\in N_2$. Let $n=d'd_1d_2\ldots$ be an $N_2$-factorization of $n$ and 
    $$n=d'd_1d_2\ldots=d'd_{\al_1}d_{\al_2}\ldots d_{\al_j},$$
    where $d_{\al_i}>1$ for all $1\leq i\leq j$ and $\al_1<\al_2<\cdots<\al_j$. Note that if $d_{\al_i}$ and $d'$ are greater than $1$ then they are greater than or equal $3$. By Proposition \ref{Prop5.5}, we obtain
    $$\frac{B_n}{n}=\frac{d_1+\sum_{\al\geq 2}(d_\al-1)}{d'd_1d_2\ldots}=\frac{1-j+\sum_{i=1}^jd_{\al_j}}{d'd_{\al_1}d_{\al_2}\ldots d_{\al_{j}}}$$
    By Lemma \ref{notinN2}, we have the following 4 cases.
    
    \textbf{Case 1} $d'>1, j=0$. Then $\frac{B_n}{n}=\frac{1}{d'}\leq\frac{1}{3}$.
    
    \textbf{Case 2} $d'>1, j=1$. So $\frac{B_n}{n}=\frac{d_{\al_1}}{d'd_{\al_1}}=\frac{1}{d'}\leq\frac{1}{3}$.
    
    \textbf{Case 3} $j=2$. Without loss of generality, we may assume that $d_{\al_2}\leq d_{\al_1}$. Hence, we have
    $$\frac{B_n}{n}=\frac{-1+d_{\al_1}+d_{\al_2}}{d'd_{\al_1}{d_{\al_2}}}\leq\frac{2d_{\al_1}}{d'd_{\al_1}d_{\al_2}}=\frac{2}{d'd_{\al_2}}\leq\frac{2}{d_{\al_2}}\leq\frac{2}{3}$$
    
    \textbf{Case 4} $j\geq 3$. Let $d_{\al_r}=\max_{1\leq i\leq j}d_{\al_i}$.
    \[
    \frac{B_n}{n}=\frac{1-j+\sum_{i=1}^jd_{\al_i}}{d'd_{\al_1}\ldots d_{\al_j}}\leq\frac{\sum_{i=1}^jd_{\al_i}}{d'd_{\al_1}\ldots d_{\al_j}}\leq\frac{jd_{\al_r}}{d'd_{\al_1}\ldots d_{\al_j}}\\
    =\frac{j}{d'\prod_{1\leq i\leq j,~ i\neq r}d_{\al_i}}\leq \frac{j}{\prod_{1\leq i\leq j,~ i\neq r}d_{\al_i}}.
    \]
    Let $s$ be an index such that $j-1\leq s\leq j$ and $s\neq r$. Then $j<2^{j-1}\leq 2^s\leq 2^{\al_s}$. Since $d_{\al_s}\in P_{\al_s}$, we have $d_{\al_s}\geq 2^{\al_s}+1$ by Lemma \ref{lem43}.	Hence,  $j<2^{\al_s}<d_{\al_s}$.  Therefore, 
    \[
    \frac{B_n}{n}\leq \frac{j}{\prod_{1\leq i\leq j,~ i\neq r}d_{\al_i}}\leq \frac{d_{\al_s}}{\prod_{1\leq i\leq j,~ i\neq r}d_{\al_i}}=\frac{1}{\prod_{1\leq i\leq j,~ i\neq r, i\neq s}d_{\al_i}}\leq \frac{1}{3}.
    \]
    Altogether, we obtain $B_n\leq \frac{2n}{3}$, and hence
    $$E(n)=\frac{5n}{9}-\frac{2B_n}{9}\geq \frac{5n}{9}-\frac{4n}{27}=\frac{11n}{27}.$$
\end{proof}

From Theorem \ref{final},  it can be concluded that  $E(n)$ grows at the same rate with $n$ as $n$ is odd and tends to infinity.

\section{Conclusion and Remarks}
The hulls of cyclic codes of odd lengths  over $\Z$  has been studied. The characterization of the hulls   has been given in terms of their generators.  The types  of the hulls of cyclic codes of arbitrary odd length have been determined as well. 
Subsequently, the $2$-dimension of the hulls of cyclic codes of odd length over $\Z$ has been determined together with the average $2$-dimension of the hull of cyclic codes of odd length $n$ over $\Z$.  Upper and lower  bounds for  the average $2$-dimension of the hull have been given.  Asymptotically, it has been shown that the average of $2$-dimension of the hull of cyclic codes of  odd length over $\Z$ grows  the same rate as the length of the  codes.

It would be interesting to study the  properties the hulls of  cyclic codes of even lengths over $\Z$. An extension of this paper to the case of the hulls of cyclic or constacyclic codes over finite chain rings is an  interesting research  problem as well. 

\section*{Acknowledgements}  The authors would like to thank the anonymous referees for their helpful comments.
This research supported by the Thailand Research Fund  and the Office of Higher Education Commission of Thailand under  Research
Grant MRG6080054.

\end{document}